\documentclass[runningheads, envcountsame, a4paper]{llncs} 
 
\usepackage{microtype}

\usepackage{amssymb}
\usepackage{graphicx,stackrel,amsmath,comment,tikz,xspace,booktabs,multirow}

\newtheorem{innercustomthm}{Theorem}
\newenvironment{customthm}[1]
  {\renewcommand\theinnercustomthm{#1}\innercustomthm}
  {\endinnercustomthm}
  
  \newtheorem{innercustomlem}{Lemma}
\newenvironment{customlem}[1]
  {\renewcommand\theinnercustomlem{#1}\innercustomlem}
  {\endinnercustomlem}

\newtheorem{innerproposition}{Proposition}
\newenvironment{customprop}[1]
  {\renewcommand\theinnerproposition{#1}\innerproposition}
  {\endinnerproposition}

\newcommand{\rfig}[1]{Fig.~\ref{#1}}
\newcommand{\rlem}[1]{Lemma~\ref{#1}}

\newcommand{\rthm}[1]{Theorem~\ref{#1}}

\newcommand{\req}[1]{Equation~(\ref{#1})}

\newcommand{\sltwoz}{SL$(2,\mathbb{Z})$\xspace}
\newcommand{\vect}[1]{\mathbf{#1}}
\newcommand{\sd}{\stackbin[]{*}{\Rightarrow}}

\usetikzlibrary{spy,shapes.geometric}

\usetikzlibrary{intersections,automata}

\tikzset{elliptic state/.style={draw,ellipse}}

\usepackage{url}
\urldef{\mailsa}\path|{sangkiko,potapov}@liverpool.ac.uk|
\newcommand{\keywords}[1]{\par\addvspace\baselineskip
\noindent\keywordname\enspace\ignorespaces#1}

\title{Matrix Semigroup Freeness Problems in \sltwoz\thanks{This research was supported by EPSRC grant EP/M00077X/1.}}

\author{Sang-Ki Ko \and Igor Potapov}
\institute{Department of Computer Science, University of Liverpool\\
Ashton Street, Liverpool, L69 3BX, United Kingdom\\
\mailsa
}

\begin{document}

\maketitle

\begin{abstract}
In this paper we study decidability and complexity of decision problems on matrices from the special linear group \sltwoz. In particular, we study the freeness problem: given a finite set of matrices $G$ generating a multiplicative semigroup $S$, decide whether each element of $S$ has at most one factorization over $G$. In other words, is $G$ a code? We show that the problem of deciding 
whether a matrix semigroup in \sltwoz is non-free is NP-hard.
Then, we study questions about the number of factorizations of matrices in the matrix semigroup such as the 
finite freeness problem, the recurrent matrix problem, the unique factorizability problem, etc.
Finally, we show that some factorization problems could be even harder in \sltwoz, for example we show 
that to decide whether every prime matrix has at most $k$ factorizations is PSPACE-hard.
%
 \end{abstract}

\section{Introduction}

In general, many computational problems for matrix semigroups are proven to be undecidable starting from 
dimension three or four
\cite{BellP08,Identity,CassaigneHK99,HalavaHH07,Potapov2004}. 
One of the central decision problems for matrix semigroups is the membership problem.
Let $S = \langle G \rangle$ be a matrix semigroup generated by a generating set~$G$. 
The {\em membership problem} is to decide whether or not a given matrix $M$ belongs to the matrix semigroup $S$. 
In other words the question is whether a matrix $M$ can be factorized over the generating set~$G$ or not. 

Another fundamental problem for matrix semigroups is the {\em freeness problem},
where we want to know whether every matrix in the matrix semigroup has a unique factorization 
over $G$.
Mandel and Simon~\cite{MandelS77} showed that the freeness problem is decidable in polynomial 
time for matrix semigroups with a single generator for any dimension over rational numbers. 
Indeed, the freeness problem for matrix semigroups with a single generator is the complementary 
problem of the {\em matrix torsion problem} which asks whether there exist two integers $p,q \ge 1$ 
such that $M^p = M^{q+p}$. Klarner et al.~\cite{KlarnerBS91} proved that the freeness problem 
in dimension three over natural numbers is undecidable. 

Decidability of the freeness problem in dimension two has been already an 
open problem for a long time~\cite{BCK2004,CassaigneHK99}. However the solutions for
some special cases exist. For example Charlier and Honkala~\cite{CharlierH14} 
showed that the freeness problem is decidable for 
upper-triangular matrices in dimension two over rationals 
when the products are restricted to certain bounded languages.
Bell and Potapov~\cite{BellP082} showed that the freeness problem is 
undecidable in dimension two for matrices over quaternions.

The study in~\cite{CassaigneHK99} revealed a class of matrix semigroups 
formed by two $2\times2$ matrices over natural numbers for which the 
freeness in unknown, highlighting a particular pair:
\[
\begin{pmatrix} 2 & 0 \\ 0 & 3 \end{pmatrix} \mbox{ and }\begin{pmatrix} 3 & 5 \\ 0 & 5 \end{pmatrix}.
\]

The above case was simultaneously shown to be non-free in two papers~\cite{CassaigneN12} and~\cite{GawrychowskiGK10}, where authors were providing new algorithms for checking freeness at some subclasses.
However the status of the freeness problem for natural, integer and complex numbers is still unknown.
The decidability of the freeness problem for \sltwoz was shown in~\cite{CassaigneN12} following the idea of 
solving the membership problem in \sltwoz shown in~\cite{ChoffrutK10}.

The effective symbolic representation of matrices in \sltwoz leads recently to several 
decidability and complexity results. The mortality, identity and vector reachability problems were shown to be NP-hard for \sltwoz in~\cite{BellHP12,BellP12}. For the modular group, the membership problem was shown to be decidable in polynomial time by Gurevich and Schupp~\cite{GurevichS07}.
Decidability of the membership problem in matrix semigroups in \sltwoz and the {\em identity problem}
 in $\mathbb{Z}^{2\times 2}$ was shown to be decidable in \cite{ChoffrutK10} in 2005.
Later in 2016, Semukhin and Potapov showed that the {\em vector reachability problem} is 
also decidable in \sltwoz~\cite{PotapovS16}. 

In this paper we study decidability and complexity questions related to freeness and various other 
factorization problems in \sltwoz . The new hardness results are interesting in the context 
of understanding complexity in matrix semigroups in general and the decidability results on factorizations in \sltwoz
can be important in other areas and fields. In particular, the special linear group~\sltwoz 
has been extensively exploited in
hyperbolic geometry~\cite{ElstrodtGM88,Zagier08}, dynamical systems~\cite{POLTEROVICHR04}, Lorenz/modular knots~\cite{Mackenzie09}, braid groups~\cite{Potapov13}, high energy physics~\cite{Witten05}, M/en theories~\cite{MoralMPR11}, music theory~\cite{Noll07}, and so on.

In this paper, we show that the question about non-freeness for matrix semigroups in \sltwoz 
is NP-hard by finding a different reduction than the one used in \cite{BellHP12,BellP12}.
Then we show both decidability and hardness results for the {\em finite freeness problem}: decide whether or not every matrix in the matrix semigroup has a finite number of factorizations. Also we prove NP-hardness of the problem whether 
a given matrix has more than one factorization in \sltwoz and undecidability of this problem in $\mathbb{Z}^{4\times 4}$, or more specifically in SL($4,\mathbb{Z}$). Then it is shown that both problems whether a particular matrix has an infinite number factorizations 
or it has more than $k$ factorizations are decidable and NP-hard in \sltwoz while they are undecidable in $\mathbb{Z}^{4\times 4}$.
Finally we show that some of the factorizations problems could be even harder in \sltwoz, for example 
we show that to decide whether every prime matrix has at most $k$-factorizations is PSPACE-hard.

\section{Preliminaries}

In this section we formulate several problems, provide important definitions and notation as well
as several technical lemmas used throughout the paper.\\


\noindent{\bf Basic definitions.}
A {\em semigroup} is a set equipped with an associative binary operation. Let $S$ 
be a semigroup and $X$ be a subset of $S$. We say that a semigroup $S$ is {\em generated} 
by a subset $X$ of $S$ if each element of $S$ can be expressed as a composition of 
elements of $X$. Then, we call $X$ the {\em generating set} of $S$.
 Then, $X$ is a {\em code} if and only 
if every element of $S$ has a unique factorization over $X$. A semigroup~$S$ is {\em free}
if there exists a subset~$X \subseteq S$ which is a code and $S = X^+$.

Given an alphabet $\Sigma = \{1,2, \ldots, m\}$, a word $w$ is an element of $\Sigma^*$. 
For a letter $a \in \Sigma$, we denote by $\overline{a}$ the inverse letter of $a$ 
such that $a \overline{a} = \varepsilon$ where $\varepsilon$ is the empty word.

A {\em nondeterministic finite automaton} (NFA) is a tuple
$A = (\Sigma, Q, \delta, q_0, F)$ where $\Sigma$ is
the input alphabet, $Q$ is the finite set of states,
$\delta \colon Q \times \Sigma \rightarrow 2^Q$ is
the multivalued transition function, $q_0 \in Q$ is the
initial state and $F \subseteq Q$ is the set of final states.
In the usual way $\delta$ is extended as a function
$Q \times \Sigma^* \rightarrow 2^Q$ and the language accepted
by $A$ is $L(A) = \{ w \in \Sigma^* \mid \delta(q_0, w) \cap F
\neq \emptyset \}$. The automaton $A$ is a  deterministic finite
automaton (DFA) if $\delta$ is a single valued  function $Q \times \Sigma \rightarrow Q$.
It is well known that the deterministic and nondeterministic
finite automata recognize the class of {\em regular languages\/}
\cite{Shallit08}.\\

\noindent{\bf Factorization and freeness problems.}
Let $S$ be a matrix semigroup generated by a finite set~$G$ of matrices. Then we 
define a matrix~$M$ is {\em $k$-factorizable} for $k \in \mathbb{N}$ if there 
are at most $k$ different factorizations of $M$ over $G$. In the matrix semigroup freeness 
problem, we check whether every matrix in $S$ is $1$-factorizable.

\begin{problem}
Given a finite set~$G$ of $n \times n$ matrices generating a matrix semigroup~$S$, 
is $S$ free? (i.e., does every element $M \in S$ have a unique factorization over $G$?)
\end{problem}

The above problem is well-known as the {\em freeness problem}. Clearly, the 
{\em non-freeness problem} is to decide whether the matrix semigroup $S$ 
is not free.

For a matrix~$M$, if there exists $k < \infty$ where $M$ is $k$-factorizable, 
then we say that $M$ is {\em finitely factorizable}. In other words, a finitely 
factorizable matrix~$M$ has finitely many different factorizations over~$G$.
We define a matrix semigroup~$S$ is {\em finitely free} if every matrix in 
$S$ is finitely factorizable and define the {\em finite freeness problem} as follows:

\begin{problem}
Given a finite set~$G$ of $n \times n$ matrices generating a matrix semigroup~$S$, 
does every element $M \in S$ have a finite number of factorizations over~$G$?
\end{problem}

Freeness and finite freeness problems are asking about factorization properties 
for all matrices in the semigroup. In case where a semigroup is 
not free or not finitely free, instead of asking whether the semigroup 
is free or finitely free, it is possible to
define problems for a given particular matrix in the semigroup as follows: 

\begin{problem}
Given a finite set~$G$ of $n \times n$ matrices generating a matrix semigroup~$S$ 
and a matrix~$M$ in $S$, does $M$ have
\begin{enumerate}
\item a unique factorization over $G$? (matrix unique factorizability problem)
\item at most $k$ factorizations over $G$? (matrix $k$-factorizability problem)
\item an infinite number of factorizations over $G$? (recurrent matrix problem)
\end{enumerate}
\end{problem}

\usetikzlibrary{arrows,calc}
\tikzset{
>=stealth',
help lines/.style={dashed, thick},
axis/.style={<->},
important line/.style={thick},
connection/.style={thick, dotted},
}

\noindent{\bf Group alphabet encodings.}
Let us introduce several technical lemmas that will be used 
in encodings for NP-hardness and undecidability results. Our original encodings require the use of 
group alphabet and the following lemmas for showing the transformation from an arbitrary
group alphabet into a binary group alphabet and later into matrix form that 
is computable in polynomial time.
%
%
\begin{lemma}\label{lem:binaryencoding}
Let $\Sigma = \{ z_1, z_2, \ldots, z_l \}$ be a group alphabet and $\Sigma_2 = \{ a,b, \overline{a}, \overline{b} \}$ be a binary group alphabet. Define the mapping $\alpha : \Sigma \to \Sigma_2^* $ 
by:
\[
\alpha(z_i) = a^i b \overline{a}^i,\;\; \alpha(\overline{z_i}) = a^i \overline{b} \overline{a}^i,
\]
where $1 \le i \le l$. Then $\alpha$ is a monomorphism. Note that $\alpha$ can be extended to domain 
$\Sigma^*$ in the usual way.
\end{lemma}

\begin{lemma}[Lyndon and Schupp~\cite{LyndonS77}]\label{lem:matrixencoding}
Let $\Sigma_2 = \{ a,b, \overline{a}, \overline{b} \}$ be a binary group alphabet and define $f : \Sigma_2^* 
\to \mathbb{Z}^{2\times 2}$ by:
\[
f(a) = \begin{pmatrix}1 & 2 \\ 0 & 1 \end{pmatrix}, f(\overline{a}) = \begin{pmatrix}1 & -2 \\ 0 & 1 \end{pmatrix}, f(b) = \begin{pmatrix}1 & 0 \\ 2 &1 \end{pmatrix}, f(\overline{b}) = \begin{pmatrix}1 & 0 \\ -2 & 1 \end{pmatrix}.
\]
Then $f$ is a monomorphism.
\end{lemma}

The composition of Lemmas~\ref{lem:binaryencoding} and \ref{lem:matrixencoding} gives 
us the following lemma that ensures that encoding the {\em subset sum problem} (SSP) and 
the {\em equal subset sum problem} (ESSP) instances into matrix semigroups can be completed 
in polynomial time.

\begin{lemma}[Bell and Potapov~\cite{BellP12}]
Let $z_j$ be in $\Sigma$ and $\alpha,f$ be mappings as defined in~Lemmas~\ref{lem:binaryencoding} and \ref{lem:matrixencoding}, 
then, for any $i \in \mathbb{N}$,
\[
f(\alpha(z_j^i)) = f((a^j b \overline{a}^j)^i) = \begin{pmatrix} 1+ 4ij & -8ij^2 \\ 
2i & 1 - 4ij \end{pmatrix}.
\]
\end{lemma}

\noindent{\bf  Symbolic representation of matrices from \sltwoz .}
Here we provide another technical details about the representation of \sltwoz and their 
properties~\cite{BellHP16,Rankin77}.
It is known that \sltwoz is generated by two matrices 
\[\vect{S} = \begin{pmatrix} 0 & -1 \\ 1 & 0 \end{pmatrix} \mbox{ and }
\vect{R} = \begin{pmatrix} 0 & -1 \\ 1 & 1 \end{pmatrix},\]
which have respective orders 4 and 6. This implies that every matrix in \sltwoz is a product of 
$\vect{S}$ and $\vect{R}$. Since $\vect{S}^2 = \vect{R}^3 = -\vect{I}$, every matrix in \sltwoz 
can be uniquely brought to the following form:
\begin{equation}\label{eq:unique}
(-\vect{I})^{i_0} \vect{R}^{i_1} \vect{S} \vect{R}^{i_2} \vect{S} \cdots \vect{S} \vect{R}^{i_{n-1}} \vect{S} \vect{R}^{i_n},
\end{equation}
where $i_0 \in \{0,1\}$, $i_1,i_n \in \{0,1,2\}$, and $i_j \ne 0 \mod 3$ for $1 < j < n$.

The representation (\ref{eq:unique}) for a given matrix in  \sltwoz  is unique, but it is very common
to present this result ignoring the sign, i.e. considering the projective special linear group.
%
%
Let $\Sigma_{SR} = \{ s,r\}$ be a binary alphabet. We define a mapping~$\varphi: \Sigma_{SR} \to 
{\rm SL}(2, \mathbb{Z})$ as follows: $\varphi(s) = \vect{S}$ 
and $\varphi(r) = \vect{R}$.
Naturally, we can 
extend the mapping $\varphi$ to the morphism~$\varphi : \Sigma_{SR}^* \to {\rm SL}(2, \mathbb{Z})$.
We call a word~$w \in \Sigma_{SR}^*$ {\em reduced} if there is no occurrence of subwords~$ss$ 
or $rrr$ in $w$. Then, we have the following fact.

\begin{theorem}[Lyndon and Schupp~\cite{LyndonS77}]\label{thm:Lyndon}
For every matrix~$M \in {\rm SL}(2,\mathbb{Z})$, there exists a unique reduced word~$w \in \Sigma_{SR}^*$ 
in form of (\ref{eq:unique})
such that either $M = \varphi(w)$ or $M = - \varphi(w)$.
\end{theorem}

\noindent
Following \rthm{thm:Lyndon}, all word representations 
of  a particular matrix~$M$  in \sltwoz
over the alphabet 
$\Sigma_{SR}$  can be expressed as a 
context-free language.

\begin{lemma}\label{lem:context}
Let $M$ be a matrix in \sltwoz. Then, there exists a context-free language over $\Sigma_{SR}$ 
which contains all representations~$w \in \Sigma_{SR}^*$ such that $\varphi(w) = M$.
\end{lemma}

\section{Matrix semigroup freeness}

The matrix semigroup freeness problem is to determine whether every matrix in 
the semigroup has a unique factorization. Note that the decidability of the matrix 
semigroup freeness in \sltwoz has been shown by 
Cassaigne and Nicolas~\cite{CassaigneN12} but the complexity of the problem 
has not been resolved yet despite various NP-hardness results on other matrix problems~\cite{BellHP12,BellP12}.
Here we show that the problem of deciding whether the matrix semigroup in \sltwoz is not free is NP-hard by encoding 
different NP-hard problem comparing to the one used in~\cite{BellHP12,BellP12}.

\begin{theorem}
Given a matrix semigroup~$S$ in \sltwoz generated by the set~$G$ of matrices, the problem of deciding whether $S$ is not free is NP-hard.
\end{theorem}
\begin{proof}
We use an encoding of the {\em equal subset sum problem} (ESSP), which is proven to be NP-hard, 
into a set of two-dimensional 
integral matrices~\cite{WoegingerY92}. The ESSP is, given a set~$U = \{ s_1, s_2, \ldots, s_k \}$ of 
$k$ integers, to decide whether or not there exist two disjoint nonempty subsets 
$U_1, U_2 \subseteq U$ whose elements sum up to the same value. Namely,
$\sum_{s_1 \in U_1} s_1 = \sum_{s_2 \in U_2} s_2.$

Define an alphabet
$\Sigma = \{0,1,  \ldots,  k-1,  \overline{1},\overline{2}, \ldots, \overline{(k-2)}, \overline{(k-1)}, \overline{k}, 
a\}.$
We define a set~$W$ of words which encodes the ESSP instance.
\[W = \{ i \cdot a^{i+1} \cdot \overline{(i+1)},\;\; i \cdot \varepsilon \cdot \overline{(i+1)}  \mid 0 \le i \le k-1 \} \subseteq \Sigma^*.
\]


We define `border letters' as letters from $\Sigma \setminus \{ a \}$ and the inner border letters 
of a word as all border letters excluding the first and last. 
We call a word a `partial cycle' if all inner border 
letters in that word are inverse to a consecutive inner border letter. Moreover, we note that for any partial cycle~$u \in W^+$ 
the first border letter of $u$ is strictly smaller than the last border letter if we compare them as integers.
\rfig{fig:structure} shows the structure of our encoding of the ESSP instance.

First we prove that if there is a solution to the ESSP instance, then the matrix semigroup generated 
by matrices encoded from 
the set~$W$ is not free.
Let us assume that there exists a solution to the ESSP instance, which is two sequences of integers 
where each of two sequences sums up to the same integer~$x$. Then, the solution can be represented 
by the following pair of sequences:
\[
Y = (y_1, y_2, \ldots, y_{k-1}, y_k) \mbox{  and  }Z = (z_1, z_2, \ldots, z_{k-1}, z_k),
\]
where $y_i, z_i \in \{ 0, s_i \}, 1 \le i \le k$ and $\sum_{i=1}^k y_i = \sum_{i=1}^k z_i = x$. Note that 
$y_i \ne z_i$ in at least one index $i$ for $1 \le i \le k$.

For a sequence~$Y$, there exists a word~$w_Y = w_1 w_2 \cdots w_{k} \in W^+$ such that 
$w_i = (i-1) \cdot a^{y_i} \cdot \overline{i}$.
Since $\sum_{i=1}^k y_i = x$, the reduced representation of~$w_Y$ is $r(w_Y) = 0 \cdot s^x \cdot \overline{k}$ as all inner border letters are cancelled. Analogously, we have a word~$w_Z$ for a sequence~$Z$ and its 
reduced representation~$r(w_Z)$ is equal to $r(w_Y)$ as the sum of integers in the sequence~$Z$ is equal 
to the sum of integers in~$Y$.
As we have two words in $W^+$ whose reduced representations are equal, 
the semigroup generated by matrices encoded from the set~$W$ is not free.

Now we prove the opposite direction: if there is no solution to the ESSP instance, then the matrix semigroup 
is free.
Assume that there is no solution to the ESSP instance and the matrix semigroup is not free. Since the 
matrix semigroup is not free, we have two different words $w, w' \in W^+$ whose reduced representations 
are equal, namely, $r(w) = r(w')$.

\begin{figure}[t]
\centering
\begin{tikzpicture}[->,>=stealth',shorten >=1pt,auto,node distance=1.7cm,
                    semithick]

  \node[elliptic state] (A)                  {0};
  \node[elliptic state]         (B) [right of=A] {1};
  \node[elliptic state]         (C) [right of=B] {2};
    \node[]         (Z) [right of=C] {$\cdots$};
  \node[elliptic state]         (D) [right of=Z] {$k-2$};
    \node[elliptic state]         (E) [right of=D] {$k-1$};
        \node[elliptic state]         (F) [right of=E] {$k$};

  \path (A) edge     [bend left]         node {$a^{s_1}$} (B)
   edge     [bend right]         node {$\varepsilon$} (B)
(B) edge     [bend left]         node {$a^{s_2}$} (C)
   edge     [bend right]         node {$\varepsilon$} (C)
(C) edge     [bend left,dotted]         node {$$} (Z)
   edge     [bend right,dotted]         node {$$} (Z)
(Z) edge     [bend left,dotted]         node {$$} (D)
   edge     [bend right,dotted]         node {$$} (D)
(D) edge     [bend left]         node {$a^{s_{k-1}}$} (E)
   edge     [bend right]         node {$\varepsilon$} (E)
(E) edge     [bend left,above]         node {$a^{s_{k}}$} (F)
   edge     [bend right,above]         node {$\varepsilon$} (F);
\end{tikzpicture}
\caption{Structure of the matrix semigroup encoded by the set~$W$. Each matrix in the generating set 
of the matrix semigroup corresponds to each transition of the automaton structure.}
\label{fig:structure}
\end{figure}
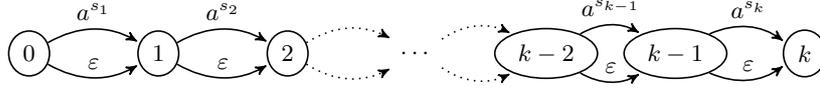

For a word~$w$, we decompose~$w$ into subwords~$w = u_1 u_2 \cdots u_m$ such that each $u_i 
\in W^+, 1 \le i \le m$ is a partial cycle of maximal size. Similarly, we decompose $w'$ into subwords 
of maximal partial cycles as follows: $w' = u_1' u_2' \cdots u_n'$. 
Since $r(w) = r(w')$, it follows that 
$r(u_i) = r(u_i')$ should hold for $1 \le i \le m$ and $m=n$. On the other hand, since $w \ne w'$, there 
exists $i, 1\le i \le m$ where $u_i \ne u_i'$. Note that the maximal partial cycles $u_i$ and $u_i'$ should 
have the same number of $a$'s since $r(u_i) = r(u_i')$ and the letter $a$ cannot be cancelled by the
reduction of words.
As we mentioned earlier, the first border letter and last border letter of a partial cycle are 
integers where the first border letter is strictly smaller than the last border letter. Let us say that 
$i_1$ is the first border letter and $i_2$ is the last border letter of $u_i$ and $u_i'$. Then, 
the number of $a$'s in $u_i$ and $u_i'$ is the sum of subset of integers from the set~$\{ s_{i_1+1}, 
s_{i_1+2}, \ldots, s_{i_2} \}$. It follows from the fact that $u_i \ne u_i'$ that we have two distinct 
subsets of the set~$\{ s_{i_1+1}, s_{i_1+2}, \ldots, s_{i_2} \}$ whose sums are the same. 
This contradicts our assumption since we have two disjoint subsets of equal subset sum.
\qed\end{proof}

Recently, Bell et al. proved that the problem of deciding whether the identity matrix 
is in $S$, where $S$ is an arbitrary regular subset of \sltwoz, is in NP~\cite{BellHP16}. 
Since we can show that the matrix semigroup~$S$ is not free by showing that the 
equation $M_1 M M_2 = M_3 M' M_4$ is satisfied where 
$M_1 \ne M_3$, $M_2 \ne M_4$, and $M_i, M, M' \in S$ for $1 \le i \le 4$. We can 
show that $S$ is not free by showing that the matrix
$M_1 M M_2 M_4^{-1} M'^{-1} M_3^{-1}$ is the identity matrix.

Let $M_1 M^* M_2 M_4^{-1} (M^{-1})^* M_3^{-1}$ be a regular subset of \sltwoz 
subject to $M_1 \ne M_3$, $M_2 \ne M_4$ and $M \in S$. Then, we can decide 
whether or not $S$ is free by deciding whether or not a regular subset of \sltwoz 
contains the identity matrix. Therefore, we can conclude as follows:

\begin{corollary}
Given a matrix semigroup~$S$ in \sltwoz generated by the set~$G$ of matrices, the problem of deciding whether $S$ is not free is NP-complete.
\end{corollary}

If the matrix semigroup is not free (not every matrix have unique factorization)
we still have a question whether each matrix in a given semigroup has only a finite number of factorizations.
Next we show that the problem of checking whether there exists a matrix in the 
semigroup which has an infinite number of factorizations is 
decidable and NP-hard in \sltwoz.

\begin{theorem}\label{thm:ffp}
Given a matrix semigroup~$S$ in \sltwoz generated by the set~$G$ of matrices, the problem of deciding whether $S$ contains 
a matrix with an infinite number of factorizations is decidable and NP-hard.
\end{theorem}
\begin{proof}
Let us consider a matrix semigroup~$S$ which is generated by the set 
$G = \{ M_1,M_2, \ldots, M_n\}$ of matrices.
Let $w_1, w_2, \ldots, w_n \in \Sigma_{SR}^*$ be words encoding the generators, such that 
$\varphi(w_i) = M_i$ for $1 \le i \le n$. Then, we can define a regular language~$L_S$ 
corresponding to $S$ as $L_S = \{ w_1, w_2 ,\ldots , w_n\}^+$.
Let $A = (Q, \Sigma, \delta, Q_0, F)$ be an NFA accepting $L_S$ constructed based on $S$. 
For states~$q$ and $p$, where the state~$p$ is reachable from $q$ by reading 
$ss$ or $rrr$, we add an $\varepsilon$-transition from $q$ to $p$. We repeat 
this process until there is no such pair of states following to the procedure proposed in~\cite{ChoffrutK10}.

If there exists a matrix $M$ which can be represented 
by infinitely many factorizations over $G$, then there is an 
infinite number of accepting runs for the matrix $M$ in $A$. 
It is easy to see that we have an infinite number of accepting runs 
for some matrix $M \in S$ if and only if there is a cycle only consisting of 
$\varepsilon$-transitions. As we can compute the $\varepsilon$-closure of 
states in $A$, the problem of deciding whether there exists a matrix 
with an infinite number of factorizations is decidable.

For the NP-hardness of the problem, we modify and adapt 
the NP-hardness proof of the identity problem in \sltwoz~\cite{BellP12}. 
We use an encoding of the {\em subset sum problem} (SSP), 
which is, given a set~$U = \{ s_1, s_2, \ldots, s_k \}$ of 
$k$ integers, to decide whether or not there exists a subset
$U' \subseteq U$ whose elements sum up to the given integer~$x$. Namely,
$\sum_{s \in U'} s = x.$

Define an alphabet
$\Sigma = \{0,1,  \ldots,  2k+1, \overline{1}, \overline{2}, \ldots,  \overline{(2k+1)}, a,b,\overline{a}, \overline{b}\}.$
We define a set~$W$ of words which encodes the SSP instance.
\begin{align}\label{eq:ssp}
\begin{split}
W = &\;\; \{ i \cdot a^{i+1} \cdot \overline{(i+1)},\;\; i \cdot \varepsilon \cdot \overline{(i+1)}  \mid 0 \le i \le k-1 \}\;\; \cup\\
       & \;\; \{ i \cdot b^{i+1} \cdot \overline{(i+1)},\;\; i \cdot \varepsilon \cdot \overline{(i+1)}  \mid k+1 \le i \le 2k \} \;\;\cup\\
       &\;\; \{ k \cdot \overline{a}^x \cdot \overline{(k+1)} \} \cup \{(2k+1) \cdot \overline{b}^{x} \cdot \overline{0}\} \subseteq \Sigma^*.
\end{split}       
\end{align}

\rfig{fig:structure2} shows the structure of the word encoding of the SSP instance.
The full proof for showing that the matrix semigroup $S$ corresponding to $W^+$ has 
a matrix with an infinite number of factorizations if and only if the SSP instance has a 
solution can be found in the archive version of the paper.
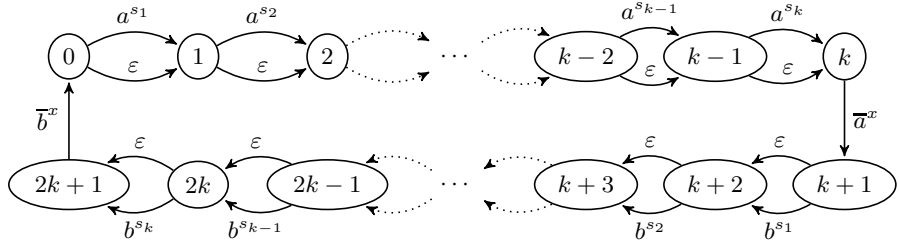
\begin{figure}[t]
\centering
\begin{tikzpicture}[->,>=stealth',shorten >=1pt,auto,node distance=1.7cm,
                    semithick]

\node[elliptic state]         (A) {0};
\node[elliptic state]         (B) [right of=A] {1};
\node[elliptic state]         (C) [right of=B] {2};
\node[]         (C1) [right of=C] {$\cdots$};
\node[elliptic state]         (D) [right of=C1] {$k-2$};
\node[elliptic state]         (E) [right of=D] {$k-1$};
\node[elliptic state]         (F) [right of=E] {$k$};
\node[elliptic state]         (G) [below of=F] {$k+1$};
\node[elliptic state]         (H) [left of=G] {$k+2$};
\node[elliptic state]         (I) [left of=H] {$k+3$};
\node[]         (C2) [left of=I] {$\cdots$};
\node[elliptic state]         (J) [left of=C2] {$2k-1$};
\node[elliptic state]         (K) [left of=J] {$2k$};
\node[elliptic state]         (L) [left of=K] {$2k+1$};

  \path (A) edge     [bend left]         node {$a^{s_1}$} (B)
   edge     [bend right]         node {$\varepsilon$} (B)
(B) edge     [bend left]         node {$a^{s_2}$} (C)
   edge     [bend right]         node {$\varepsilon$} (C)
(C) edge     [bend left,dotted]         node {$$} (C1)
   edge     [bend right,dotted]         node {$$} (C1)
   (C1) edge     [bend left,dotted]         node {$$} (D)
   edge     [bend right,dotted]         node {$$} (D)
(D) edge     [bend left]         node {$a^{s_{k-1}}$} (E)
   edge     [bend right]         node {$\varepsilon$} (E)
(E) edge     [bend left,above]         node {$a^{s_{k}}$} (F)
   edge     [bend right,above]         node {$\varepsilon$} (F)
(F) edge              node {$\overline{a}^{x}$} (G)
(G) edge     [bend left]         node {$b^{s_1}$} (H)
   edge     [bend right,above]         node {$\varepsilon$} (H)
   (H) edge     [bend left]         node {$b^{s_2}$} (I)
   edge     [bend right,above]         node {$\varepsilon$} (I)
   (C2) edge     [bend left,dotted]         node {$$} (J)
   edge     [bend right,dotted]         node {$$} (J)
   (I) edge     [bend right,dotted]         node {$$} (C2)
   edge     [bend left,dotted]         node {$$} (C2)
   (J) edge     [bend left,below]         node {$b^{s_{k-1}}$} (K)
   edge     [bend right,above]         node {$\varepsilon$} (K)
   (K) edge     [bend left,below]         node {$b^{s_k}$} (L)
   edge     [bend right,above]         node {$\varepsilon$} (L)
(L) edge              node {$\overline{b}^{x}$} (A)
;
\end{tikzpicture}
\caption{Structure of the matrix semigroup encoded by the set~$W$.}
\label{fig:structure2}
\end{figure}
\qed\end{proof}

\section{Matrix factorizability problems}

In the matrix semigroup freeness problem, we ask whether every matrix 
in the semigroup has a unique factorization. Instead of considering a question about 
every matrix in the semigroup, we restrict our question to a given particular matrix,
which may have a unique factorization, a finite number of unique factorizations or even an infinite
number of unique factorizations.

\subsection{Unique factorizability problem}

In the {\em matrix unique factorizability problem}, 
we consider the problem of deciding 
whether or not a particular matrix $M$ in $S$ has a unique factorization 
over $G$. 
We first establish the decidability and NP-hardness of the problem.

\begin{theorem}\label{thm:uniquefp}
Given a matrix semigroup~$S$ in \sltwoz generated by the set~$G$ of 
matrices and a particular matrix~$M$ in $S$, the problem of deciding whether 
the matrix~$M$ has more than one factorization over~$G$ is decidable and NP-hard.
\end{theorem}

\begin{proof}
From \rlem{lem:context}, we can represent a set of all unreduced representations 
for $M$ over $\Sigma_{SR} = \{s,r\}$ as a context-free language~$L_M$.

We can also obtain a regular language that corresponds to the 
matrix semigroup~$S$. Let $G = \{M_1, M_2, \ldots, M_n\}$ be the generating set 
of $S$. Namely, $S = \langle M_1, M_2, \ldots, M_n \rangle$.
Let $w_1, w_2, \ldots, w_n \in \Sigma_{SR}^*$ be words encoding the generators, such that 
$\varphi(w_i) = M_i$ for $1 \le i \le n$. Then, we can define a regular language~$L_S$ 
corresponding to $S$ as $L_S = \{ w_1, w_2 ,\ldots , w_n\}^+$.
Then, the intersection of $L_M \cap L_S$ contains all words that correspond to the matrix~$M$ 
in the semigroup~$S$. If the cardinality of 
$L_M \cap L_S$ is larger than one, we immediately have two different factorizations 
for the matrix~$M$ over $G$. Therefore, let us assume that $|L_M \cap L_S| = 1$ and 
$w$ be the only word in $L_M \cap L_S$. Clearly, $\varphi(w) = M$ 
and $M$ can be generated by the set~$G$. Note that each accepting path of $w$ in $L_S$ 
corresponds to a unique factorization of $M$ over $G$. Now we can decide whether 
or not $M$ has a unique factorization over $G$ by counting the number of accepting paths of words in 
$L_M \cap L_S$ from an NFA accepting $L_S$. 

The NP-hardness can be proven by the reduction from the SSP 
in a similar manner to the proof of \rthm{thm:ffp}.
See \req{eq:ssp} for the word encoding of the SSP instance. 
Let us pick the word~$w = 0 \cdot \varepsilon \cdot \overline{1}$ in $W$ and notice that the matrix $M = f(\alpha(w))$ which 
is encoded from $w$ is in the matrix semigroup $S$.
We will show that the 
matrix~$M$ in $S$ has at least two factorizations over 
the generating set~$\{ f(\alpha(w)) \mid w \in W \}$ of $S$ if and only if the SSP instance 
has a solution. The full proof can be found in the archive version.
\qed\end{proof}

We reduce the {\em fixed element PCP} (FEPCP)~\cite{BellP08} which is proven to be undecidable 
to the unique factorizability problem over $\mathbb{Z}^{4\times 4}$ for the following 
undecidability result.

\begin{theorem}\label{thm:unde}
Given a matrix semigroup~$S$ over $\mathbb{Z}^{4\times 4}$ 
generated by the set~$G$ of matrices and a particular matrix~$M$ in $S$, 
the problem of deciding whether the matrix~$M$ has more than one factorization over~$G$ is undecidable.
\end{theorem}

\subsection{Recurrent matrix problem}

We first tackle the problem of 
deciding whether or not a particular matrix in the semigroup has an infinite number 
of factorizations. Note that we call this decision problem the {\em recurrent matrix problem} 
instead of the {\em matrix finite factorizability problem} as we named for the other variants. 
The recurrent matrix problem has been introduced by Bell and Potapov~\cite{BellP08} 
and proven to be undecidable for matrices over $\mathbb{Z}^{4 \times 4}$ based on the reduction from 
FEPCP.

We show that the recurrent matrix problem is decidable and NP-hard for matrix semigroups in 
\sltwoz. We first mention that the recurrent matrix problem is different with the 
identity problem. One may think that the recurrent matrix problem is equivalent to the 
identity problem since it is obvious that if the identity matrix exists then every matrix in 
the semigroup has an infinite number of factorizations. However, the opposite does not 
hold as follows:

\begin{proposition}\label{prop:recurrent}
Let $S$ be a matrix semigroup generated by the generating set~$G$ and 
$M$ be a matrix in $S$. Then, the matrix $M$ has an infinite number of factorizations over $G$ 
if the identity matrix exists in $S$. However, the opposite does not hold in general.
\end{proposition}

Now we establish the results for the recurrent matrix problem in \sltwoz.

\begin{theorem}\label{thm:recurrent}
The recurrent matrix problem in \sltwoz is decidable and in fact, NP-hard.
\end{theorem}

%

We also consider the matrix $k$-factorizability problem which is to decide 
whether a particular matrix $M$ in the semigroup 
has at most $k$ factorizations over the generating set $G$.

\begin{lemma}\label{lem:kfreeness}
Given a matrix semigroup~$S$ in \sltwoz generated by the set~$G$ of 
matrices, a particular matrix~$M \in S$, and a positive integer $k \in \mathbb{N}$, the problem of deciding whether 
the matrix~$M$ has more than $k$ factorizations over~$G$ is decidable and NP-hard.
\end{lemma}

%

%
%

We mention that the matrix $k$-factorizability problem is also undecidable 
over $\mathbb{Z}^{4\times 4}$ following \rthm{thm:unde}.

\section{On the finite number of factorizations}

Recall that the matrix semigroup freeness problem 
examines whether or not there exists a matrix in the semigroup has 
more than one factorization. The finite freeness problem asks whether 
there exists a matrix in the semigroup which has an infinite number of factorizations. 
In that sense, we may interpret these problems as the problems asking whether 
the number of factorizations in the semigroup is bounded by one (the freeness problem) 
or unbounded (the finite freeness problem).

In this section, we are interested in finding a number~$k \in \mathbb{N}$ 
by which the number of factorizations of matrices in the matrix semigroup 
is bounded. In other words, we check whether every matrix in the semigroup 
is $k$-factorizable. However, it is not easy to define the {\em $k$-freeness problem} 
as we define the general freeness problem by the following observation.

%
%
%

Let $S$ be a matrix semigroup generated by the set~$G$ of matrices and
$M$ be a $k$-factorizable matrix over $G$. Let us denote the number of 
factorizations of $M$ by ${\rm dec}(M)$. Thus, we can write ${\rm dec}(M) = k$. 
It is easy to see that $S$ 
is free if for every matrix~$M$ in $S$, ${\rm dec}(M) = 1$. Let us assume that 
${\rm dec}(M_1) = m$ and ${\rm dec}(M_2) = n$ for 
$m,n \in \mathbb{N}$.
Then, ${\rm dec}(M_1 M_2) = k$ where $k \ge mn$. This means that 
if $S$ is not free, then there is no finite value $k$ such 
that every matrix in $S$ is $k$-factorizable.

In that reason, we define the following notion which prevents the multiplicative 
property of the number of factorizations.
We say that a matrix $M$ is {\em prime} if it is impossible to decompose $M$ into 
$M = M_1 M_2$ such that 
${\rm dec}(M) = {\rm dec}(M_1) \times {\rm dec}(M_2)$, ${\rm dec}(M_1) \ne 1$, and ${\rm dec}(M_2) \ne 1$. 
We define a matrix semigroup~$S$ to be {\em $k$-free} if 
every prime matrix~$M$ in $S$ has at most $k$ different factorizations over~$G$. 
Formally, a matrix semigroup $S$ is $k$-free if and only if 
$\max\{ {\rm dec}(M) \mid M \in S, \;\; \mbox{$M$ is prime} \} \le k.$

This definition 
gives rise to the following problem which is a generalized version of the matrix semigroup 
freeness problem.

\begin{problem}
Given a finite set~$G$ of $n \times n$ matrices generating a matrix semigroup~$S$, 
does every prime element $M \in S$ have at most $k$ factorizations over~$G$?
\end{problem}

In this paper, we leave the decidability of the $k$-freeness problem open but establish the 
PSPACE-hardness result as a lower bound of the problem, which is interesting compared 
to the NP-hardness of the other freeness problems.
	
\begin{theorem}\label{thm:kfreenesshard}
Given a matrix semigroup~$S$ in \sltwoz generated by the set~$G$ of matrices 
and a positive integer $k \in \mathbb{N}$, 
the problem of deciding whether or not 
every prime matrix in $S$ has at most $k$ factorizations is PSPACE-hard.
\end{theorem}
\begin{proof}
For the PSPACE-hardness of the problem, we reduce the DFA intersection emptiness 
problem~\cite{Kozen77} 
to the $k$-freeness problem. Note that given $k$ DFAs, 
the DFA intersection emptiness problem asks whether the intersection of $k$ DFAs is 
empty. The full proof can be found in the Appendix.
\qed\end{proof}

\section{Conclusions}

We have investigated the matrix semigroup freeness problems. The freeness 
problem is to decide whether or not every matrix in the given matrix semigroup 
has a unique factorization over the generating set of matrices. The freeness 
problem was already known to be decidable in \sltwoz. Here we have shown the 
that the freeness problem in \sltwoz is NP-hard which, along with the fact that 
the problem is in NP~\cite{BellHP16}, proves that the freeness problem in \sltwoz 
is NP-complete.
We also have studied a relaxed variant called the finite freeness problem 
in which we decide whether or not every matrix in the semigroup has a finite 
number of factorizations. We prove that the finite freeness problem in \sltwoz 
is decidable and NP-hard.

Moreover, we have considered the problem on the number of factorizations 
leading to a given particular matrix in the semigroup. The matrix unique factorizability 
problem asks whether a given matrix in the semigroup has a unique factorization 
over the generating set. We have proven that the problem is decidable and 
NP-hard. We also have studied the recurrent matrix problem that decides whether a 
particular matrix in the semigroup has an infinite number of factorizations and 
shown that it is decidable and NP-hard as well.

Lastly, we have examined the $k$-freeness problem which is a problem of deciding 
whether every prime matrix in the matrix semigroup has at most $k$ factorizations.
We have proven the PSPACE-hardness for the $k$-freeness problem which implies that the 
problem is computationally more difficult than the general freeness problem. 
We also have established the decidability and NP-hardness of the 
matrix $k$-factorizability in \sltwoz.

\bibliographystyle{abbrv}
\bibliography{freeness}

\begin{thebibliography}{10}

\bibitem{BellHP12}
P.~C. Bell, M.~Hirvensalo, and I.~Potapov.
\newblock Mortality for $2\times2$ matrices is {NP}-hard.
\newblock In {\em Proceedings of the 37th International Symposium on
  Mathematical Foundations of Computer Science}, pages 148--159, 2012.

\bibitem{BellHP16}
P.~C. Bell, M.~Hirvensalo, and I.~Potapov.
\newblock The identity problem for matrix semigroups in {SL}$(2,\mathbb{Z})$ is
  {NP}-complete.
\newblock 2016.
\newblock To appear in SODA 17.

\bibitem{BellP08}
P.~C. Bell and I.~Potapov.
\newblock Periodic and infinite traces in matrix semigroups.
\newblock In {\em Proceedings of the 34th Conference on Current Trends in
  Theory and Practice of Computer Science}, pages 148--161, 2008.

\bibitem{BellP082}
P.~C. Bell and I.~Potapov.
\newblock Reachability problems in quaternion matrix and rotation semigroups.
\newblock {\em Information and Computation}, 206(11):1353--1361, 2008.

\bibitem{Identity}
P.~C. Bell and I.~Potapov.
\newblock On the undecidability of the identity correspondence problem and its
  applications for word and matrix semigroups.
\newblock {\em Int. J. Found. Comput. Sci.}, 21(6):963--978, 2010.

\bibitem{BellP12}
P.~C. Bell and I.~Potapov.
\newblock On the computational complexity of matrix semigroup problems.
\newblock {\em Fundamenta Infomaticae}, 116(1-4):1--13, 2012.

\bibitem{BCK2004}
V.~D. Blondel, J.~Cassaigne, and J.~Karhum\"{a}ki.
\newblock Problem 10.3: Freeness of multiplicative matrix semigroups.
\newblock In {\em Unsolved Problems in Mathematical Systems and Control
  Theory}, pages 309--314. Princeton University Press, 2004.

\bibitem{CassaigneHK99}
J.~Cassaigne, T.~Harju, and J.~Karhum\"{a}ki.
\newblock On the undecidability of freeness of matrix semigroups.
\newblock {\em International Journal of Algebra and Computation},
  09(03n04):295--305, 1999.

\bibitem{CassaigneN12}
J.~Cassaigne and F.~Nicolas.
\newblock On the decidability of semigroup freeness.
\newblock {\em RAIRO - Theoretical Informatics and Applications},
  46(3):355--399, 8 2012.

\bibitem{CharlierH14}
E.~Charlier and J.~Honkala.
\newblock The freeness problem over matrix semigroups and bounded languages.
\newblock {\em Information and Computation}, 237:243--256, 2014.

\bibitem{ChoffrutK10}
C.~Choffrut and J.~Karhum\"{a}ki.
\newblock Some decision problems on integer matrices.
\newblock {\em RAIRO - Theoretical Informatics and Applications},
  39(1):125--131, 3 2010.

\bibitem{ElstrodtGM88}
J.~Elstrodt, F.~Grunewald, and J.~Mennicke.
\newblock Arithmetic applications of the hyperbolic lattice point theorem.
\newblock {\em Proceedings of the London Mathematical Society},
  s3-57(2):239--283, 1988.

\bibitem{MoralMPR11}
M.~P. {Garc\'{i}a del Moral}, I.~Mart\'{i}n, J.~M. {Pe\~{n}a}, and
  A.~Restuccia.
\newblock {\rm SL}$(2, \mathbb{Z})$ symmetries, supermembranes and symplectic
  torus bundles.
\newblock {\em Journal of High Energy Physics}, (9):1--12, 2011.

\bibitem{GawrychowskiGK10}
P.~Gawrychowski, M.~Gutan, and A.~Kisielewicz.
\newblock On the problem of freeness of multiplicative matrix semigroups.
\newblock {\em Theoretical Computer Science}, 411(7-9):1115--1120, 2010.

\bibitem{GurevichS07}
Y.~Gurevich and P.~Schupp.
\newblock Membership problem for the modular group.
\newblock {\em SIAM Journal on Computing}, 37(2):425--459, 2007.

\bibitem{HalavaHH07}
V.~Halava, T.~Harju, and M.~Hirvensalo.
\newblock Undecidability bounds for integer matrices using claus instances.
\newblock {\em International Journal of Foundations of Computer Science},
  18(05):931--948, 2007.

\bibitem{KlarnerBS91}
D.~A. Klarner, J.-C. Birget, and W.~Satterfield.
\newblock On the undecidability of the freeness of integer matrix semigroups.
\newblock {\em International Journal of Algebra and Computation},
  01(02):223--226, 1991.

\bibitem{Kozen77}
D.~Kozen.
\newblock Lower bounds for natural proof systems.
\newblock In {\em Proceedings of the 18th Annual Symposium on Foundations of
  Computer Science}, pages 254--266, 1977.

\bibitem{LyndonS77}
R.~C. Lyndon and P.~E. Schupp.
\newblock {\em Combinatorial group theory}.
\newblock 1977.

\bibitem{Mackenzie09}
D.~Mackenzie.
\newblock {\em A new twist in knot theory}, volume~7.
\newblock 2009.

\bibitem{MandelS77}
A.~Mandel and I.~Simon.
\newblock On finite semigroups of matrices.
\newblock {\em Theoretical Computer Science}, 5(2):101--111, 1977.

\bibitem{Noll07}
T.~Noll.
\newblock Musical intervals and special linear transformations.
\newblock {\em Journal of Mathematics and Music}, 1(2):121--137, 2007.

\bibitem{POLTEROVICHR04}
L.~Polterovich and Z.~Rudnick.
\newblock Stable mixing for cat maps and quasi-morphisms of the modular group.
\newblock {\em Ergodic Theory and Dynamical Systems}, 24:609--619, 2004.

\bibitem{Potapov2004}
I.~Potapov.
\newblock {\em From Post Systems to the Reachability Problems for Matrix
  Semigroups and Multicounter Automata}, pages 345--356.
\newblock Springer Berlin Heidelberg, Berlin, Heidelberg, 2005.

\bibitem{Potapov13}
I.~Potapov.
\newblock {Composition Problems for Braids}.
\newblock In {\em IARCS Annual Conference on Foundations of Software Technology
  and Theoretical Computer Science}, volume~24, pages 175--187, 2013.

\bibitem{PotapovS16}
I.~Potapov and P.~Semukhin.
\newblock Vector reachability problem in {SL}$(2, \mathbb{Z})$.
\newblock In {\em 41st International Symposium on Mathematical Foundations of
  Computer Science}, pages 84:1--84:14, 2016.

\bibitem{Rankin77}
R.~Rankin.
\newblock {\em Modular Forms and Functions}.
\newblock Cambridge University Press, 1977.

\bibitem{Shallit08}
J.~Shallit.
\newblock {\em A Second Course in Formal Languages and Automata Theory}.
\newblock Cambridge University Press, New York, NY, USA, 1 edition, 2008.

\bibitem{Witten05}
E.~Witten.
\newblock {\em {\rm SL}$(2, \mathbb{Z})$ action on three-dimensional conformal
  field theories with abelian symmetry}, volume~2, pages 1173--1200.
\newblock 2005.

\bibitem{WoegingerY92}
G.~J. Woeginger and Z.~Yu.
\newblock On the equal-subset-sum problem.
\newblock {\em Information Processing Letters}, 42(6):299--302, 1992.

\bibitem{Zagier08}
D.~Zagier.
\newblock {\em Elliptic Modular Forms and Their Applications}, pages 1--103.
\newblock 2008.

\end{thebibliography}

\newpage

\section*{Appendix}

\begin{customlem}{\ref{lem:context}.}
Let $M$ be a matrix in \sltwoz. Then, there exists a context-free language over $\Sigma_{SR} = \{s,r\}$ 
which contains all unreduced representations~$w \in \Sigma_{SR}^*$ such that $\varphi(w) = M$.
\end{customlem}
\begin{proof}
By~\rthm{thm:Lyndon}, we know that there exists a unique reduced word~$w \in \Sigma_{SR}^*$ 
such that either $M = \varphi(w)$ or $M = -\varphi(w)$.
The word~$w$ can be written as $w_1 w_2 \ldots w_n$, where $w_i \in \Sigma_{SR}$ for 
$1 \le i \le n$.

Let us remind that a {\em context-free grammar} (CFG)~$G$ is a four-tuple
$G = (V, \Sigma, R, S)$, where $V$ is a set of variables, 
$\Sigma$ is a set of terminals, $R \subseteq V \times (V\cup \Sigma)^*$ 
is a finite set of productions and $S\in V$ is the start variable.
Let $\alpha A \beta$ be a word over $V \cup \Sigma$, where $A \in V$ and $A \to \gamma \in R$. 
Then, we say that A can be rewritten as $\gamma$ and 
the corresponding derivation step is denoted
 $\alpha A \beta \Rightarrow \alpha \gamma \beta$.
The reflexive, transitive closure of $\Rightarrow$ is denoted by $\sd$ and
 the context-free language generated
 by $G$ is  $L(G) = \{w \in \Sigma^* \mid S \sd w\}$.

Let $G_M = (V,\Sigma_{SR}, P, V_S)$ be a CFG, where $V = \{V_S, A^+, A^-\}$ is a finite set of nonterminals, 
$\Sigma_{SR}= \{s,r\}$ is a binary alphabet, $P$ is a finite set of production rules, and $V_S$ is the start nonterminal. We define 
$P$ to contain the following production rules:
\begin{itemize}
\item $V_S \to A_1 w_1A_2 w_2 A_3 \ldots A_n w_n A_{n+1}$,
\item $A^+ \to \varepsilon \mid sA^- s \mid rA^+rA^+ r \mid rA^- rA^-r \mid A^-A ^- \mid A^+ A^+$. and
\item $A^- \to sA^+ s \mid rA^+rA^- r \mid rA^- rA^+r \mid A^-A ^+ \mid A^- A^+$,
\end{itemize}
where $A_i \in \{A^+, A^-\}$ for $1 \le i \le n+1$.

Note that if $M = \varphi(w)$, then there exists an even number of $A^-$'s from all $A_i$ for $1\le i \le n+1$ and 
otherwise, there exists an odd number of $A^-$'s. Then, it is easy to see that the CFG~$G_M$ 
generates all unreduced words encoding the matrix~$M$ by the morphism~$\varphi$. Formally, 
we write $L(G_M) = \{ w \mid \varphi(w) = M \}.$
Clearly, $L(G_M)$ is a context-free language.
\qed\end{proof}

\begin{customthm}{\ref{thm:ffp}.}
Given a matrix semigroup~$S$ in \sltwoz generated by the set~$G$ of matrices, the problem of deciding whether $S$ contains 
a matrix with an infinite number of factorizations is decidable and NP-hard.
\end{customthm}
\begin{proof}
Let us consider a matrix semigroup~$S$ which is generated by the set 
$G = \{ M_1,M_2, \ldots, M_n\}$ of matrices.
Let $w_1, w_2, \ldots, w_n \in \Sigma_{SR}^*$ be words encoding the generators, such that 
$\varphi(w_i) = M_i$ for $1 \le i \le n$. Then, we can define a regular language~$L_S$ 
corresponding to $S$ as $L_S = \{ w_1, w_2 ,\ldots , w_n\}^+$.
Let $A = (Q, \Sigma_{SR}, \delta, Q_0, F)$ be an NFA accepting $L_S$ constructed based on $S$. 
For states~$q$ and $p$, where the state~$p$ is reachable from $q$ by reading 
$ss$ or $rrr$, we add an $\varepsilon$-transition from $q$ to $p$. We repeat 
this process until there is no such pair of states.

If there exists a matrix $M$ which can be represented 
by infinitely many factorizations over $G$, then there are an 
infinite number of accepting runs for the matrix $M$ in $A$. 
It is easy to see that we have an infinite number of accepting runs 
for some matrix $M \in S$ if and only if there is a cycle only consisting of 
$\varepsilon$-transitions. As we can compute the $\varepsilon$-closure of 
states in $A$, the problem of deciding whether there exists a matrix 
with an infinite number of factorizations is decidable.


For the NP-hardness of the problem, we modify and adapt 
the NP-hardness proof of the identity problem in \sltwoz~\cite{BellP12}. 
We use an encoding of the {\em subset sum problem} (SSP), 
which is, given a set~$U = \{ s_1, s_2, \ldots, s_k \}$ of 
$k$ integers, to decide whether or not there exists a subset
$U' \subseteq U$ whose elements sum up to the given integer~$x$. Namely,
\[
\sum_{s \in U'} s = x.
\]

Define an alphabet
$\Sigma  = \{0,1,  \ldots,  2k+1, \ldots, \overline{1}, \overline{2}, \ldots,  \overline{(2k+1)}, a,b,\overline{a}, \overline{b}\}.$
We define a set~$W$ of words which encodes the SSP instance.
\[
\begin{split}
W = &\;\; \{ i \cdot a^{i+1} \cdot \overline{(i+1)},\;\; i \cdot \varepsilon \cdot \overline{(i+1)}  \mid 0 \le i \le k-1 \}\;\; \cup\\
       & \;\; \{ i \cdot b^{i+1} \cdot \overline{(i+1)},\;\; i \cdot \varepsilon \cdot \overline{(i+1)}  \mid k+1 \le i \le 2k \} \;\;\cup\\
       &\;\; \{ k \cdot \overline{a}^x \cdot \overline{(k+1)} \} \cup \{(2k+1) \cdot \overline{b}^{x} \cdot \overline{0}\} \subseteq \Sigma^*.
\end{split}       
\]

Here we define the set of border letters as $\Sigma \setminus \{a,b,\overline{a}, \overline{b}\}$.
We can show that the matrix semigroup $S$ contains a matrix with an infinite 
number of factorizations 
if and only if the SSP instance has a solution. Remark that it is already known 
that there exists a fully reducible word $w \in W^+$ such that 
$r(w) = \varepsilon$ if and only if the SSP instance has a solution~\cite{BellP12}. 

First we prove that $S$ is not finitely free if the SSP instance has a solution. 
As we mentioned above, if the SSP has a solution, then there exists a 
word~$w \in W^+$ that reduces to $\varepsilon$.
This means that every word in $W^+$ 
has an infinite number of factorizations over the set $W$ since we can concatenate 
a sequence of words from $W$ which reduces to an empty word~$\varepsilon$ infinitely 
many times. 

For the opposite direction, we show that 
if $S$ is not finitely free, then the SSP instance has a solution. This implies 
that if the SSP instance has no solution then $S$ must be finitely free. 
Let us suppose that the SSP has no solution. We
consider any finite word~$w$ in $W^+$ and decompose the word into 
subwords of maximal partial cycles as follows: $w = u_1 u_2 \cdots u_n$. 
Now at least one of $r(u_i)$ for $1 \le i \le n$ should have 
infinitely many factorizations over $W$. 
Let $u_i$ be a maximal partial cycle 
such that $r(u_i) = i_1 \cdot w' \cdot \overline{i_2}$, where $i_1,i_2$ are border 
letters and $w'$ is a subword over $\{ a,b, \overline{a}, \overline{b}\}$, which has an infinite number of 
factorizations over $W$. Since $u_i$ is a partial cycle, all of its 
inner border letters should be cancelled. \rfig{fig:structure2} shows 
the structure of the encoding. Since $r(u_i)$ has an infinite number 
of factorizations, we need at least one word~$y$ from $W$ which appears 
infinitely many times in $u_i$. Let us assume that $y$ starts with 
the border letter~$i$, where $0 \le i \le 2k+1$. From the structure 
of our encoding described in~\rfig{fig:structure2}, we see that the 
only way to cancel the border letter $i$ is to obtain a complete cancellation 
from the word $y$ to the word ending with the border letter $\overline{i}$. 
However, we can see that it is impossible to reach a complete cancellation 
if the SSP has no solution since we cannot completely cancel the subwords 
$\overline{a}^x$ and $\overline{b}^x$. Therefore, we prove that $S$ 
is finitely free if there is no solution to the SSP instance.

We have proven that 
the problem of deciding whether there exists a matrix with an infinite 
number of factorizations in the matrix semigroup in \sltwoz 
is NP-hard since the SSP is an NP-hard problem.
\qed\end{proof}

\begin{customthm}{\ref{thm:uniquefp}.}
Given a matrix semigroup~$S$ in \sltwoz generated by the set~$G$ of 
matrices and a particular matrix~$M$ in $S$, the problem of deciding whether 
the matrix~$M$ has more than one factorization over~$G$ is decidable and NP-hard.
\end{customthm}

\begin{proof}
From \rlem{lem:context}, we can represent a set of all unreduced representations 
for $M$ over $\Sigma_{SR} = \{s,r\}$ as a context-free language~$L_M$.

We can also obtain a regular language that corresponds to the 
matrix semigroup~$S$. Let $G = \{M_1, M_2, \ldots, M_n\}$ be the generating set 
of $S$. Namely, $S = \langle M_1, M_2, \ldots, M_n \rangle$.
Let $w_1, w_2, \ldots, w_n \in \Sigma_{SR}^*$ be words encoding the generators, such that 
$\varphi(w_i) = M_i$ for $1 \le i \le n$. Then, we can define a regular language~$L_S$ 
corresponding to $S$ as $L_S = \{ w_1, w_2 ,\ldots , w_n\}^+$.
Then, the intersection of $L_M \cap L_S$ contains all words that correspond to the matrix~$M$ 
in the semigroup~$S$. We note that the cardinality of 
$L_M \cap L_S$ should be one because otherwise we have two different factorizations 
over $G$ generating $M$. Let $w$ be the word in $L_M \cap L_S$. Clearly, $\varphi(w) = M$ 
and $M$ can be generated by the set~$G$. Note that each accepting path of $w$ in $L_S$ 
corresponds to a unique factorization of $M$ over $G$. Now we can decide whether 
or not $M$ has a unique factorization over $G$ by counting the number of accepting paths of words in 
$L_M \cap L_S$ from an NFA accepting $L_S$. 

The NP-hardness can be proven by the reduction from the SSP 
in a similar manner to the proof of \rthm{thm:ffp}.
See \req{eq:ssp} for the word encoding of the SSP instance. 

Let us pick the word~$w = 0 \cdot \varepsilon \cdot \overline{1}$ in $W$ and notice that the matrix $M = f(\alpha(w))$ which 
is encoded from $w$ is in the matrix semigroup $S$.
We will show that the 
matrix~$M$ in $S$ has at least two factorizations over 
the generating set~$\{ f(\alpha(w)) \mid w \in W \}$ of $S$ if and only if the SSP instance 
has a solution. 
We first prove that if the SSP instance has a solution, then $M$ has more than one factorization.
Recall that there exists a word~$w \in W^+$ such that $r(w) = \varepsilon$ if the SSP instance has a solution~\cite{BellP12}.
Therefore, it is not difficult to see that $M$ has more than one factorization since $S$ has an 
identity matrix and $M$ has an infinite number of factorizations.

Now we consider the opposite direction: if $M$ has more than one factorization, then 
there exists a solution to the SSP instance. 
Suppose that the SSP instance has no solution to use contradiction. 

The definition of the set $W$ ensures that the word $r(w) = 0 \cdot \varepsilon \cdot \overline{1}$
can be obtained by taking the word directly or allowing some inner cancellations since 
there is no word in $W^+$ that reduces to the empty word~$\varepsilon$ if the SSP instance 
has no solution. Suppose that there is a word~$w' \in W^+$ that reduces to the same word as $w$. 
Namely, $r(w) = r(w')$. Noe that $w'$ should start with $0 \cdot a^{s_1} \cdot \overline{1}$ 
and the remaining parts should be reduced to $1 \cdot \overline{a}^{s_1} \cdot \overline{1}$. However, 
it is impossible to completely reduce $\overline{b}^x$ since the SSP has no solution and 
therefore, $M$ has only one factorization.

We can conclude that the problem of deciding whether a particular matrix $M$ in the semigroup 
has more than one factorization is NP-hard by the reduction from the SSP.
\qed\end{proof}

\begin{customthm}{\ref{thm:unde}.}
Given a matrix semigroup~$S$ over $\mathbb{Z}^{4\times 4}$ 
generated by the set~$G$ of matrices and a particular matrix~$M$ in $S$, 
the problem of deciding whether the matrix~$M$ has more than one factorization over~$G$ is undecidable.
\end{customthm}
\begin{proof}
We use the {\em fixed element PCP} (FEPCP)~\cite{BellP08} to obtain the undecidability result of 
the matrix unique factorizability over $\mathbb{Z}^{4\times 4}$.
Given an alphabet~$\Gamma = \{a,b, a^{-1}, b^{-1}, \Delta, \Delta^{-1}, \star\}$, where 
$\Gamma \setminus \{ \star \}$ forms a free group not containing `$\star$', and a finite 
set of pairs of words over $\Gamma$,
\[
P = \{(u_1, v_1), (u_2, v_2), \ldots, (u_n, v_n) \} \subset \Gamma^* \times \Gamma^*.
\]

The FEPCP asks whether or not there exists a finite sequence of indices~$s = 
(s_1, s_2, \ldots, s_k)$ such that $u_{s_1} u_{s_2} \cdots u_{s_k} = 
v_{s_1} v_{s_2} \cdots v_{s_k} = \star$.

Let $\Sigma = \{a,b\}$ be a binary alphabet.
We use the homomorphism $f : (\Sigma \cup \overline{\Sigma})^* \to \mathbb{Z}^{2\times 2}$ 
defined in \rlem{lem:matrixencoding}.

Let $\Gamma = \{a,\overline{a}, b, \overline{b}, \Delta, \overline{\Delta}, \star \}$ and 
define a mapping $\gamma$, to encode $\Gamma$ using elements of $\varphi$, 
where $\gamma: \Gamma^* \to \mathbb{Z}^{2\times 2}$ is given by:
\begin{align*}
\gamma(\star) = f(a),&&\gamma(a) = f(bab), &&\gamma(b) = f(b^2 a b^2), &&
\gamma(\Delta) = f(b^3 a b^3),\\
&& \gamma(\overline{a}) = f(\overline{b}\overline{a}\overline{b}), &&\gamma(\overline{b}) = f(\overline{b^2} \overline{a} 
\overline{b^2}),&& \gamma(\overline{\Delta}) = f(\overline{b^3} \overline{a} \overline{b^3}).
\end{align*}

Given an instance of FEPCP $P = \{ (u_i, v_i) \mid 1 \le i \le n\},$ for each $1 \le i \le n$, 
we define the following matrices:
\[
A_i = \begin{pmatrix} \gamma(u_i) & 0 \\ 0 & \gamma(v_i) \end{pmatrix}.
\]

Here we remark that $(\varepsilon, \varepsilon) \notin \langle P \rangle$ always holds by 
the reduction process from PCP to FEPCP~\cite{BellP08}.
Then, it is easy to see that if the matrix
\[
B = \begin{pmatrix} \gamma(\star) & 0 \\ 0 & \gamma(\star) \end{pmatrix}
\]
exists in the semigroup~$S$ generated by $\{ A_1, A_2, \ldots, A_n \} \subseteq \mathbb{Z}^{4 \times 4}$,
then the FEPCP instance has a solution. Now we consider a matrix semigroup~$S'$ obtained 
from $S$ by inserting the matrix~$B$ to the generating set. In other words, the semigroup~$S$ 
is generated by the set~$G = \{ A_1, A_2, \ldots, A_n,B \} \subseteq \mathbb{Z}^{4 \times 4}$.
Then, 
we can decide whether the FEPCP instance has a solution by deciding the matrix~$B$ 
has a unique factorization over $G$. Since the existence of a solution to the FEPCP 
instance is undecidable, the matrix unique factorizability over $\mathbb{Z}^{4 \times 4}$ 
is also undecidable.
\qed\end{proof}

\begin{customprop}{\ref{prop:recurrent}.}
Let $S$ be a matrix semigroup generated by the generating set~$G$ and 
$M$ be a matrix in $S$. Then, the matrix $M$ has an infinite number of factorizations over $G$ 
if the identity matrix exists in $S$. However, the opposite does not hold in general.
\end{customprop}

\begin{proof}
Define an alphabet
$\Sigma = \{0,1,   \overline{1}, a, \overline{a}\}.$
We define a set~$W$ of words which are indeed encoding a set of matrices 
generating a matrix semigroup based on two monomorphisms $\alpha$ and $f$ as follows:
\[
W = \{ 0 \cdot a \cdot \overline{0},\;\;0 \cdot \overline{a} \cdot \overline{1}, \;\; 1 \cdot \overline{a} \cdot 
\overline{1} \} \subseteq \Sigma^*.
\]
Here, the matrix semigroup~$S$ is defined as
$ S = \{ f(\alpha(w)) \mid w \in W^+ \}.$

It is not difficult to see that there is no identity matrix in $S$ by the property of 
encodings we used. However, there exists a matrix with an infinite number of 
factorizations over the generating set~$\{ f(\alpha(w)) \mid w \in W\}$. For 
instance, the word~$0 \cdot \varepsilon \cdot \overline{1}$ can be obtained 
infinitely many times by concatenating words in $W$ as we have the 
following equation for every integer $n$:
\[
0 \cdot \varepsilon \cdot \overline{1} = (0 \cdot a \cdot \overline{0})^n (0 \cdot \overline{a} \cdot 
\overline{1}) (1 \cdot \overline{a} \cdot \overline{1})^{n-1}.
\]

As we have shown that the matrix semigroup $S$ without the identity contains 
a matrix with an infinite number of factorizations, we complete the proof.
\qed
\end{proof}

\begin{customthm}{\ref{thm:recurrent}.}
The recurrent matrix problem in \sltwoz is decidable and in fact, NP-hard.
\end{customthm}

\begin{proof}
From the proof of \rthm{thm:uniquefp}, we can see that if $L_M \cap L_S$ is infinite, 
then $M$ has an infinite number of factorizations over the generating set~$G$. Since 
the finiteness of the context-free language is decidable, the recurrent matrix problem 
in \sltwoz is also decidable.

For the NP-hardness of the recurrent matrix problem in \sltwoz, we can directly 
apply the NP-hardness proof of \rthm{thm:uniquefp}. It is not difficult to see that 
the matrix~$f(\alpha(w))$ has an infinite number of factorizations if and only if 
the encoded SSP instance has a solution.
\qed\end{proof}

\begin{customlem}{\ref{lem:kfreeness}.}
Given a matrix semigroup~$S$ in \sltwoz generated by the set~$G$ of 
matrices, a particular matrix~$M \in S$, and a positive integer $k \in \mathbb{N}$, the problem of deciding whether 
the matrix~$M$ has more than $k$ factorizations over~$G$ is decidable and NP-hard.
\end{customlem}

\begin{proof}
From the proof of \rthm{thm:uniquefp}, we can easily show that the problem of 
deciding whether or not $M$ is $k$-factorizable over $G$ is decidable.

Recall that the intersection of $L_M \cap L_S$ contains all words that correspond to the matrix~$M$ 
in the semigroup~$S$. We note that the cardinality of 
$L_M \cap L_S$ is finite if and only if $M$ is finitely factorizable. Now we can decide whether 
or not $M$ is $k$-factorizable by counting the number of accepting paths of words in 
$L_M \cap L_S$ from an NFA accepting $L_S$. Note that counting the number of 
accepting paths on a word in an NFA can be done in time polynomial in the length of the word.
\qed\end{proof}
\begin{customthm}{\ref{thm:kfreenesshard}.}
Given a matrix semigroup~$S$ in \sltwoz generated by the set~$G$ of matrices
and a positive integer $k \in \mathbb{N}$, 
the problem of deciding whether or not 
every prime matrix in $S$ has at most $k$ factorizations is PSPACE-hard.
\end{customthm}
\begin{proof}
For the PSPACE-hardness of the problem, we reduce the DFA intersection emptiness 
problem~\cite{Kozen77} 
to the $k$-freeness problem. Note that given $k$ DFAs, 
the DFA intersection emptiness problem asks whether the intersection of $k$ DFAs is 
empty.

Let us suppose that we are given $k+1$ DFAs from $A_1$ to $A_{k+1}$ as follows 
and asked whether the intersection is empty.
Let $A_i = (Q_i, \Sigma_A, \delta_i, q_{i,0}, F_i)$ be the $i$th DFA of $n$ states, where 
\begin{itemize}
\item $Q_i = \{q_{i,0},q_{i,1},\ldots, q_{i,n-1}\}$ is a finite set of states,
\item $\Sigma_A$ is an alphabet,
\item $\delta_i$ is the transition function, 
\item $q_{i,0} \in Q_i$ is the initial state, and
\item $F_i  \subseteq Q_i$ is a finite set of final state.
\end{itemize}

First we define an alphabet~$\Sigma_i$ for encoding the states of $A_i$ 
as follows:
\[
\Sigma_i = \{0,1,\ldots, n-1, \overline{0}, \overline{1}, \ldots, \overline{n-1}\}.
\]

Note that the number $k, 0 \le k \le n-1$ in $\Sigma_i$ encodes the state~$q_{i,k} 
\in Q_i$. We also define alphabets from $\Sigma_1$ to $\Sigma_{k+1}$ for
all DFAs from $A_1$ to $A_{k+1}$ analogously. Note that any pair of 
$\Sigma_i$ and $\Sigma_j$ are disjoint unless $i= j$.

Now we define an alphabet 
\[
\Sigma = \bigcup_{i=1}^{k+1} \Sigma_i \cup \{ \# \} \cup \Sigma_A
\]
 and a set~$W \subseteq \Sigma^*$ of 
words which encodes the instance of the DFA 
intersection problem as follows.
For each DFA $A_i$, we add the following words to the set~$W$:
\begin{itemize}
\item $l \cdot a \cdot \overline{m}$ for 
each transition~$q_{i,m} \in \delta(q_{i,l} , a)$, 
\item $\# \cdot \varepsilon \cdot \overline{0}$ for the initial 
state~$q_{i,0}$, and
\item $j \cdot \varepsilon \cdot \#$ for each final state~$q_{i,j} \in 
F_i$.
\end{itemize}

We can see that 
$\# \cdot w \cdot \# \in W^+$ if and only if $w \in L(A_i)$.
We add words corresponding to transitions of all DFAs from $A_1$ to $A_{k+1}$ 
analogously such that $\# \cdot w \cdot \# \in W^+$ if and only if $w \in \bigcup_{i=1}^{k+1} 
\in L(A_i)$. In other words, the set $W^+$ has a word of form $\# \cdot w \cdot \#$ 
which has a word $w$ in between $\#$ symbols if and only if any DFA from $A_1$ 
to $A_{k+1}$ has an accepting computation on the word $w$.

Let $S_W$ be the matrix semigroup generated by the set~$\{ f(\alpha(w)) \mid w \in W\}$.
We first prove that if $S_W$ is $k$-free, then the intersection of $k+1$ DFAs is empty. 
Assume that the intersection is not empty to use contradiction. This 
implies that there is a word~$w$ in $\bigcap_{i=1}^{k+1} L(A_i)$.
As we mentioned, $\# \cdot w \cdot \# \in W^+$ and the corresponding 
matrix~$M$ has $k+1$ different factorizations, which is more than $k$, over 
the generating set~$\{ f(\alpha(w)) \mid w \in W\}$. Now we reach a 
contradiction since $S_W$ is not $k$-free.

Now we prove that if the intersection of $k+1$ DFAs is empty, then the matrix semigroup 
$S_W$ is $k$-free. Assume that $S_W$ is 
not $k$-free. This implies that there exists a prime matrix~$M \in S_W$ which has more 
than $k$ different factorizations. Since $M \in S_W$, 
we have a corresponding word~$w_M \in W^+$ such that $f(\alpha(w_M)) = M$. 
We decompose~$w_M$ into subwords~$w = u_1 u_2 \cdots u_m$ such that each $u_l 
\in W^+, 1 \le l \le m$ is a partial cycle of maximal size. Note that a partial cycle~$x \cdot u \cdot \overline{y}$ 
implies that there is a path from $q_{i,x}$ to $q_{i,y}$ spelling out a word~$u$ in one of $k+1$ DFAs, say $A_i$. 
Since $M$ is prime, there can be only one partial cycle whose corresponding matrix is not $k$-factorizable.
Let $u_l = x \cdot u \cdot \overline{y} , 1\le l \le m$ be the partial cycle such that $f(\alpha(u_l))$ is not $k$-factorizable. 
If $x \ne \#$, then there are more than $k$ paths from $q_{i,x}$ to $q_{i,y}$ spelling out 
the word~$u$ in $A_i$. However, it is impossible for $A_i$ to have multiple paths labeled by the 
same word since $A_i$ is a DFA. If $x = \#$, then the only way to having more than $k$ paths 
labeled by the same word~$u$ is that all $k+1$ DFAs accept the word~$u$.

As a final note, we mention that the whole reduction process can be computed in polynomial time. 
Therefore, we prove that the $k$-freeness problem for matrix semigroups in \sltwoz is PSPACE-hard.
\qed\end{proof}

\end{document}